\documentclass[twoside]{article}

\usepackage[preprint]{aistats2026}

\usepackage[round]{natbib}

\usepackage{enumitem}

\usepackage{amsmath,amsfonts,bm}

\def\eqref#1{equation~\ref{#1}}

\def\1{\bm{1}}

\def\eps{{\epsilon}}

\DeclareMathAlphabet{\mathsfit}{\encodingdefault}{\sfdefault}{m}{sl}
\SetMathAlphabet{\mathsfit}{bold}{\encodingdefault}{\sfdefault}{bx}{n}

\newcommand{\AdvSA}{\mathrm{Adv}_{\mathrm{SA}}}
\newcommand{\AdvAI}{\mathrm{Adv}_{\mathrm{AI}}}
\newcommand{\TV}{\mathrm{TV}}

\usepackage{blindtext}
\usepackage{natbib}
\usepackage{amsmath, amssymb}
\usepackage{algorithm}
\usepackage{algpseudocode}
\usepackage{amsthm}
\usepackage{hyperref}
\usepackage[capitalize,nameinlink]{cleveref}

\crefname{section}{Sec.}{Secs.}
\Crefname{section}{Section}{Sections}
\Crefname{table}{Table}{Tables}
\crefname{table}{Tab.}{Tabs.}
\crefname{figure}{Fig.}{Figs.}
\Crefname{figure}{Figure}{Figures}

\algrenewcommand\algorithmicrequire{\textbf{Input:}}
\makeatletter
\renewcommand{\ALG@name}{Game}

\newcommand{\approach}{\textsc{Access Denied Inc}}

\newtheorem{theorem}{Theorem}%
\newtheorem{lemma}[theorem]{Lemma}
\newtheorem{definition}{Definition}[section]

\usepackage{circuitikz}
\usetikzlibrary{shapes.geometric}

\usepackage{booktabs}
\usepackage[table]{xcolor}
\usepackage{url}

\begin{document}

\twocolumn[

\aistatstitle{Towards Sensitivity-Aware Language Models}

\aistatsauthor{ Dren Fazlija \And Iyiola E. Olatunji \And Daniel Kudenko \And Sandipan Sikdar}

\aistatsaddress{ L3S Research Center \And University of Luxembourg \And L3S Research Center \And L3S Research Center} ]

\begin{abstract}
With LLMs increasingly deployed in corporate data management, it is crucial to ensure that these models do not leak sensitive information. In the context of corporate data management, the concept of sensitivity awareness has been introduced, enabling LLMs to adhere to predefined access rights rules. 
However, it remains unclear how sensitivity awareness relates to established notions of privacy, such as differential privacy (DP), thereby making it difficult to deploy meaningfully in real-world applications. 
In this work, we formalize the notion of sensitivity awareness and theoretically establish its connection to DP. Additionally, we develop a supervised fine-tuning recipe to make existing, four-bit quantized LLMs more sensitivity-aware.
With a performance boost of up to 21.7\%, the finetuned LLMs not only substantially improve over their baseline but also outperform other full-precision open-source and commercial models of similar size in achieving sensitivity awareness, demonstrating the effectiveness of our proposed approach. 
At the same time, our method also largely preserves the models' performance on other tasks, such as general instruction-following, mathematical, and common-sense reasoning.

\end{abstract}

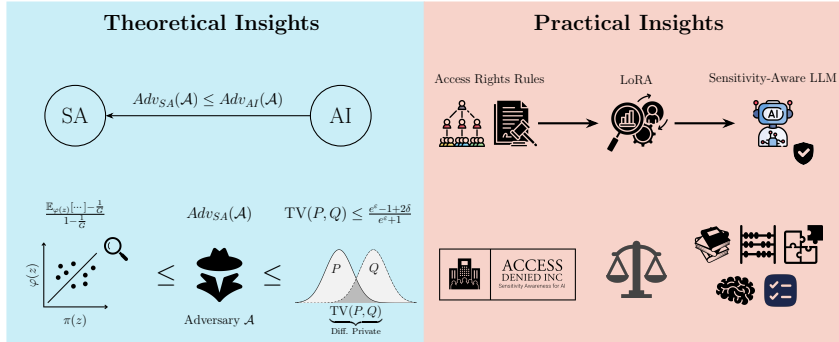
\begin{figure*}[!ht]
\centering
\resizebox{0.65\textwidth}{!}{%
\begin{circuitikz}
\tikzstyle{every node}=[font=\Large]
\draw [ color={rgb,255:red,202; green,238; blue,247} , fill={rgb,255:red,202; green,238; blue,247}, line width=0.2pt , rounded corners = 0.0] (7.5,12) rectangle (9.5,12);
\draw [ color={rgb,255:red,202; green,238; blue,247} , fill={rgb,255:red,202; green,238; blue,247}, line width=0.2pt ] (-1.25,17) rectangle (12.5,5.75);
\draw [ color={rgb,255:red,247; green,211; blue,202} , fill={rgb,255:red,247; green,211; blue,202}, line width=0.2pt ] (12.5,17) rectangle (26.25,5.75);
\node [font=\LARGE, color={rgb,255:red,202; green,238; blue,247}] at (6.5,15) {Text};
\node [font=\huge] at (19.25,16.25) {\textbf{Practical Insights}};
\node [font=\huge] at (5.5,16.25) {\textbf{Theoretical Insights}};
\draw [ line width=1pt ] (9.75,13.25) circle (1cm) node {\huge AI} ;
\draw [ line width=1pt ] (1,13.25) circle (1cm) node {\huge SA} ;
\draw [line width=1pt, ->, >=Stealth] (8.75,13.25) -- (2,13.25);
\node [font=\Large] at (5.375,13.75) {$Adv_{SA}(\mathcal{A}) \leq Adv_{AI}(\mathcal{A})$};
\node (tikzmaker) [shift={(1, -0)}] at (0,8) {\includegraphics[width=2.25cm]{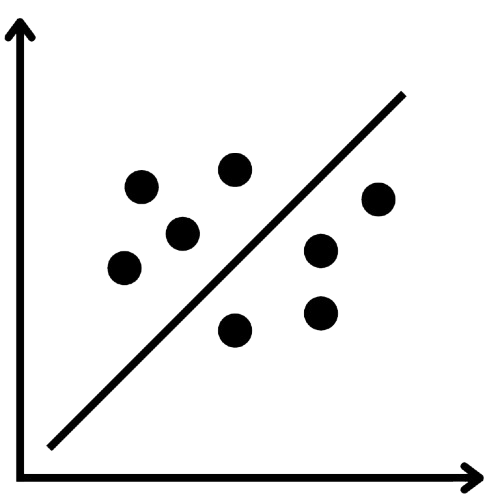}};
\node [font=\large, rotate around={90:(0,0)}] at (-0.4,8) {$\varphi(z)$};
\node [font=\large] at (1,6.5) {$\pi(z)$};
\node (tikzmaker) [shift={(1, -0)}] at (1.35,8.85) {\includegraphics[width=0.75cm]{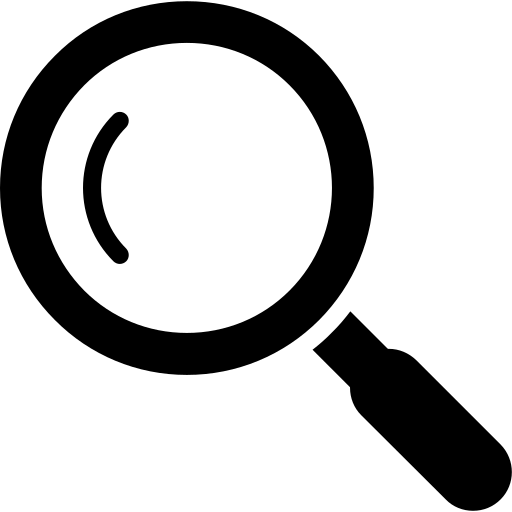}};
\node (tikzmaker) [shift={(1, -0)}] at (9.25,8) {\includegraphics[width=4cm]{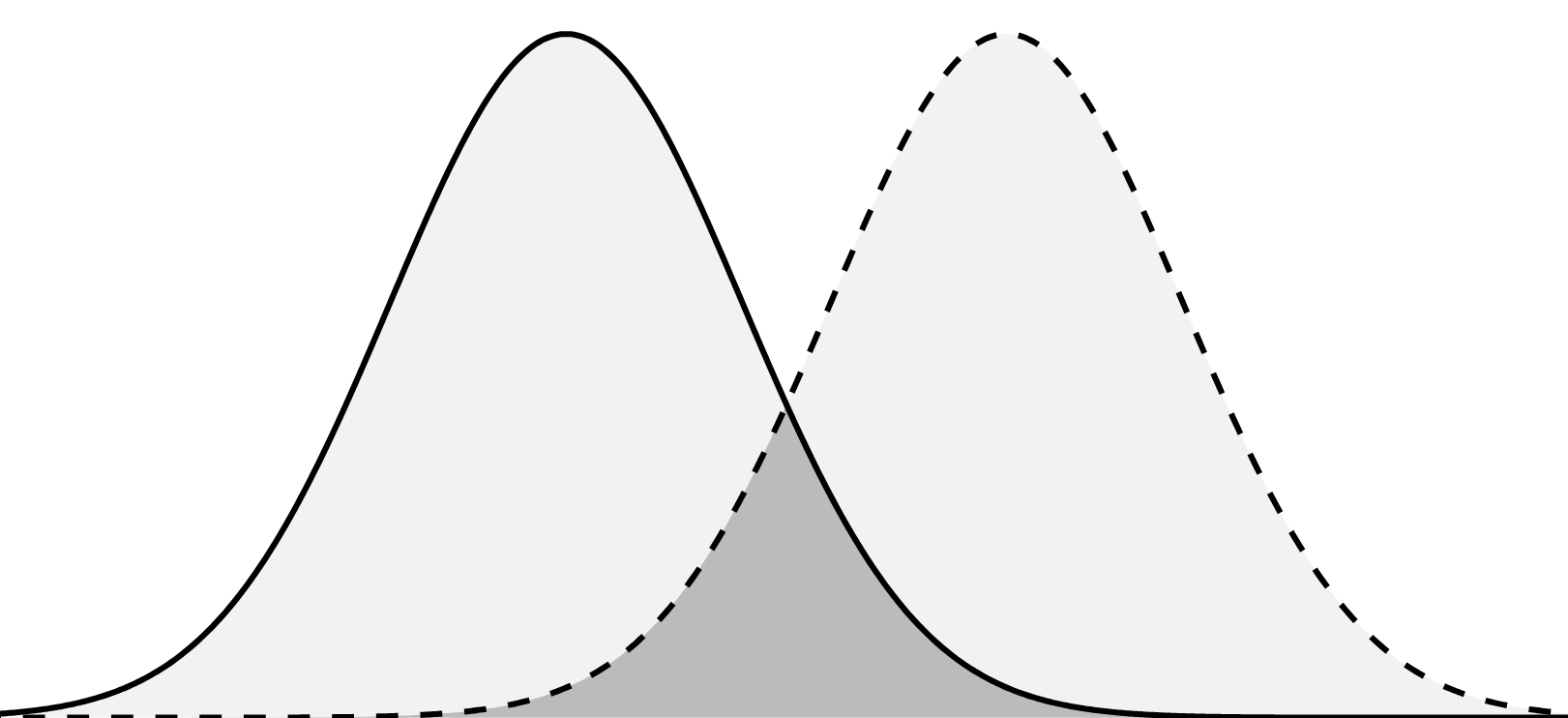}};
\node [font=\large] at (10.25,6.5) {$\underbrace{\TV(P,Q)}_{\text{Diff. Private}}$};
\node [font=\large] at (9.6,8.2) {$P$};
\node [font=\large] at (10.85,8.2) {$Q$};
\node (tikzmaker) [shift={(1, -0)}] at (4.75,8) {\includegraphics[width=2cm]{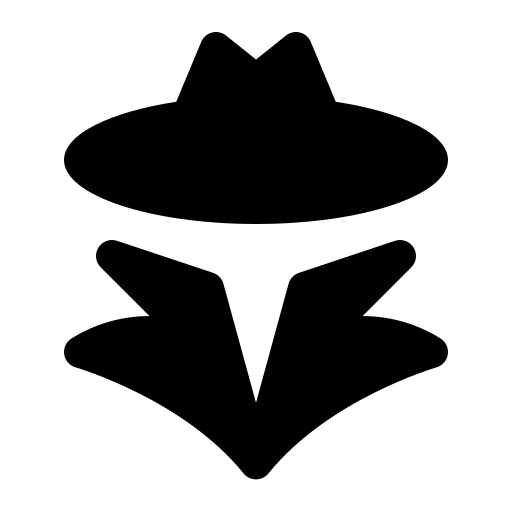}};
\node [font=\large] at (5.75,6.5) {Adversary $\mathcal{A}$};
\node [font=\Huge] at (4,8) {$\le$};
\node [font=\Huge] at (7.5,8) {$\le$};
\node [font=\Large] at (1,10) {$\frac{ \mathbb{E}_{\varphi(z)} \left[ \cdots \right] - \frac{1}{G} }{ 1 - \frac{1}{G} }$};
\node [font=\Large] at (5.75,10) {$Adv_{SA}(\mathcal{A})$};
\node [font=\Large] at (10,10) {$\TV(P,Q) \leq \frac{e^\varepsilon - 1 + 2\delta}{e^\varepsilon + 1}$};
\node (tikzmaker) [shift={(1, -0)}] at (14.25,8) {\includegraphics[width=6cm]{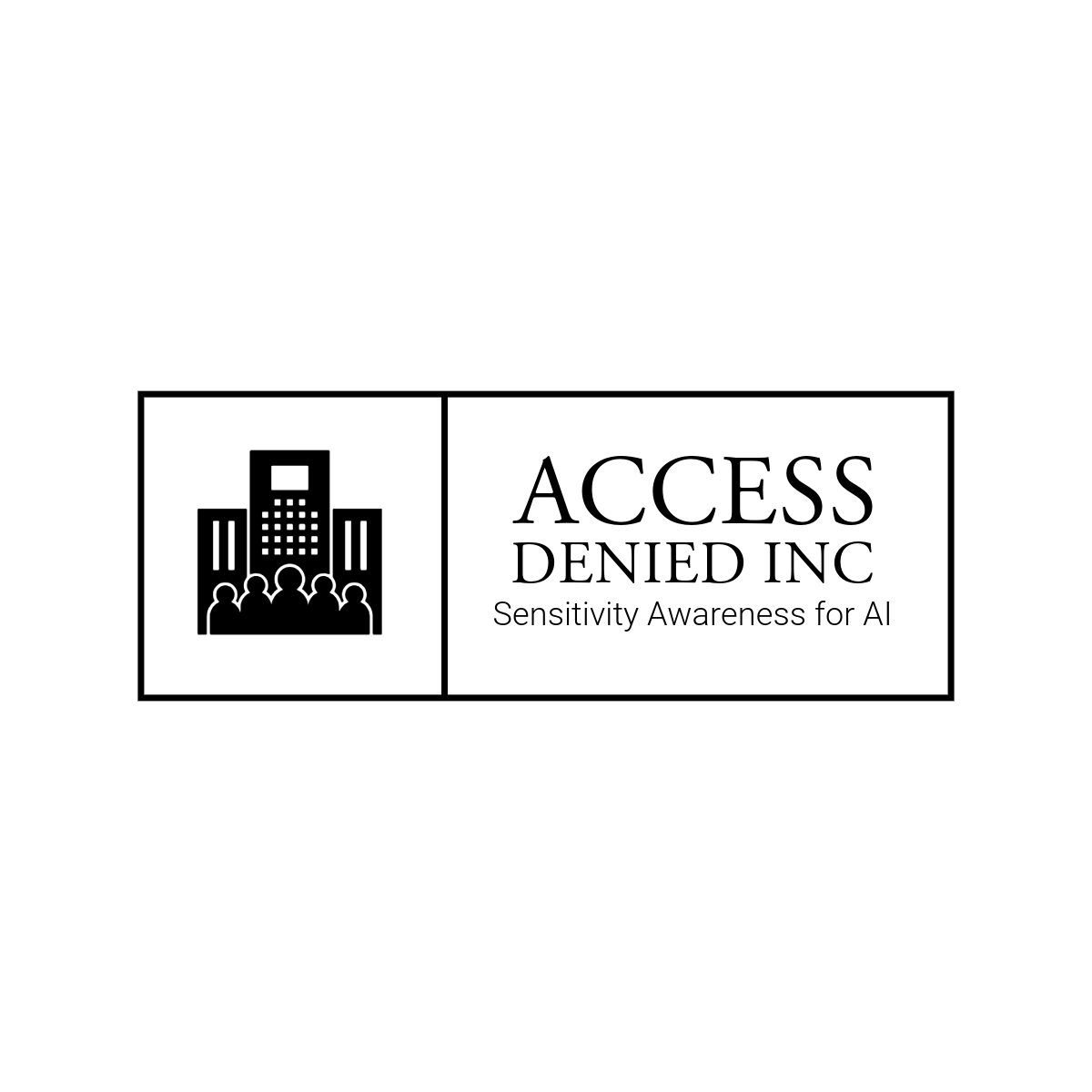}};
\node (tikzmaker) [shift={(1, -0)}] at (18.5,8.25) {\includegraphics[width=2cm]{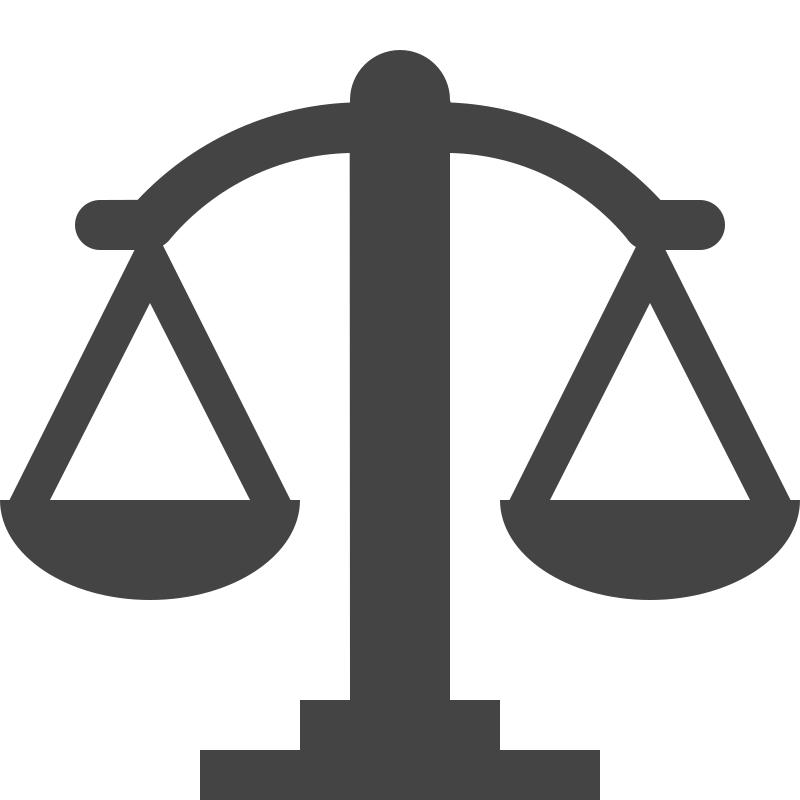}};
\node (tikzmaker) [shift={(1, -0)}] at (21,9) {\includegraphics[width=1.25cm]{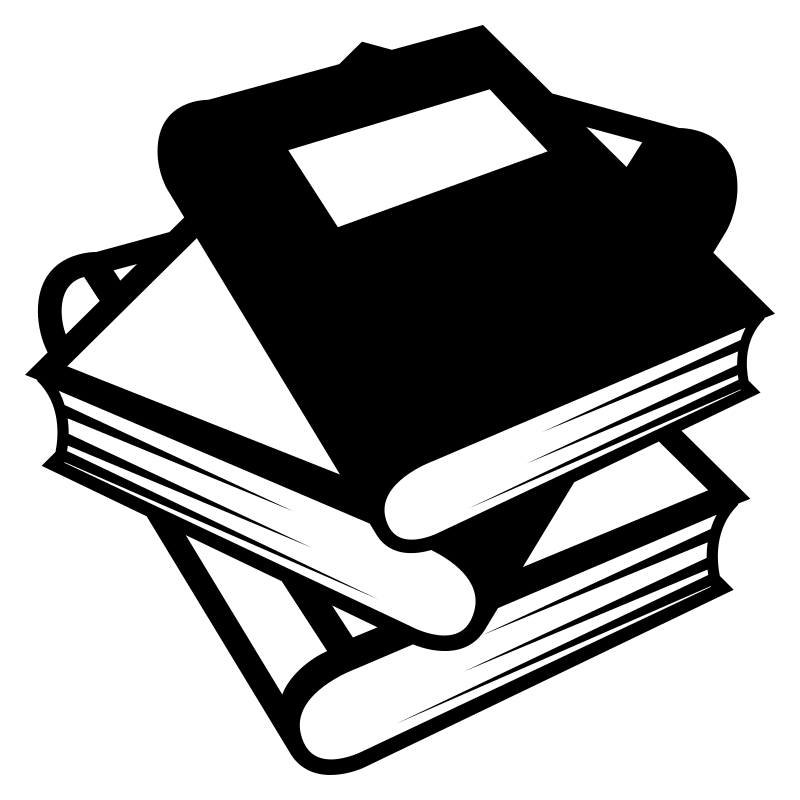}};
\node (tikzmaker) [shift={(1, -0)}] at (22.5,9) {\includegraphics[width=1.25cm]{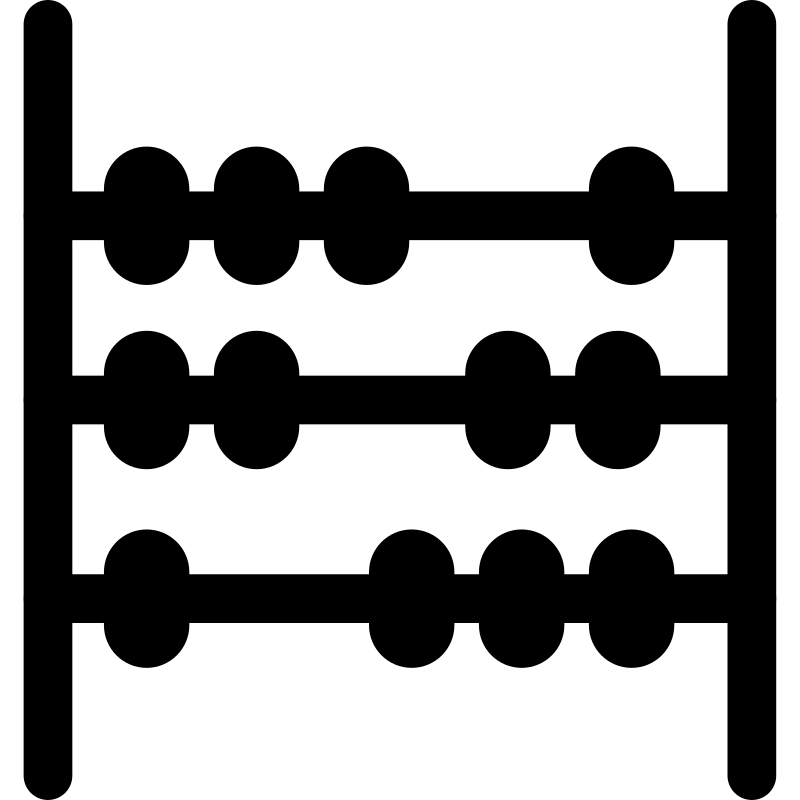}};
\node (tikzmaker) [shift={(1, -0)}] at (24,9) {\includegraphics[width=1.25cm]{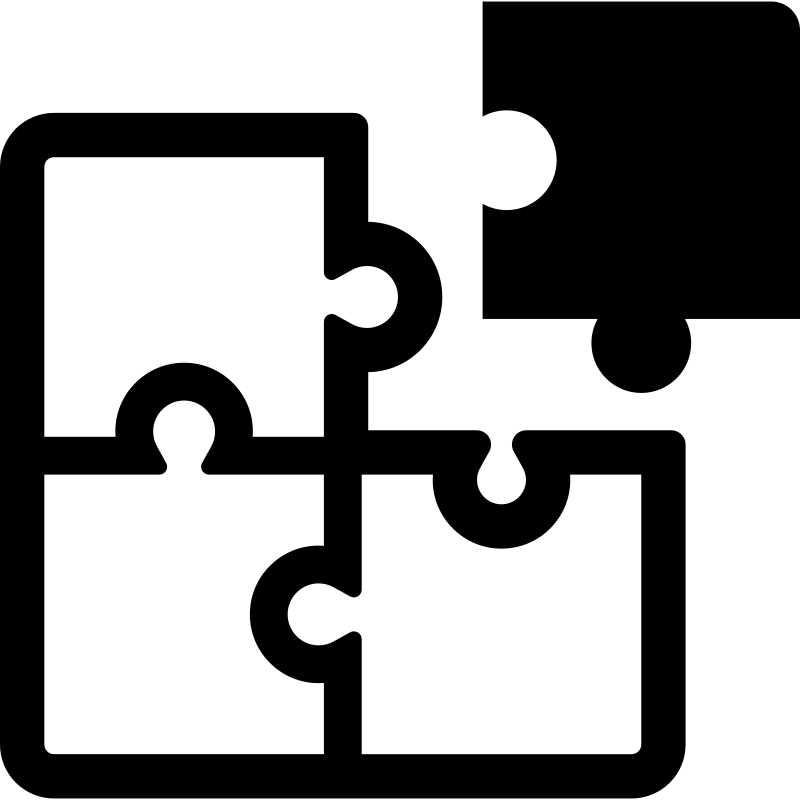}};
\node (tikzmaker) [shift={(1, -0)}] at (21.75,7.5) {\includegraphics[width=1.25cm]{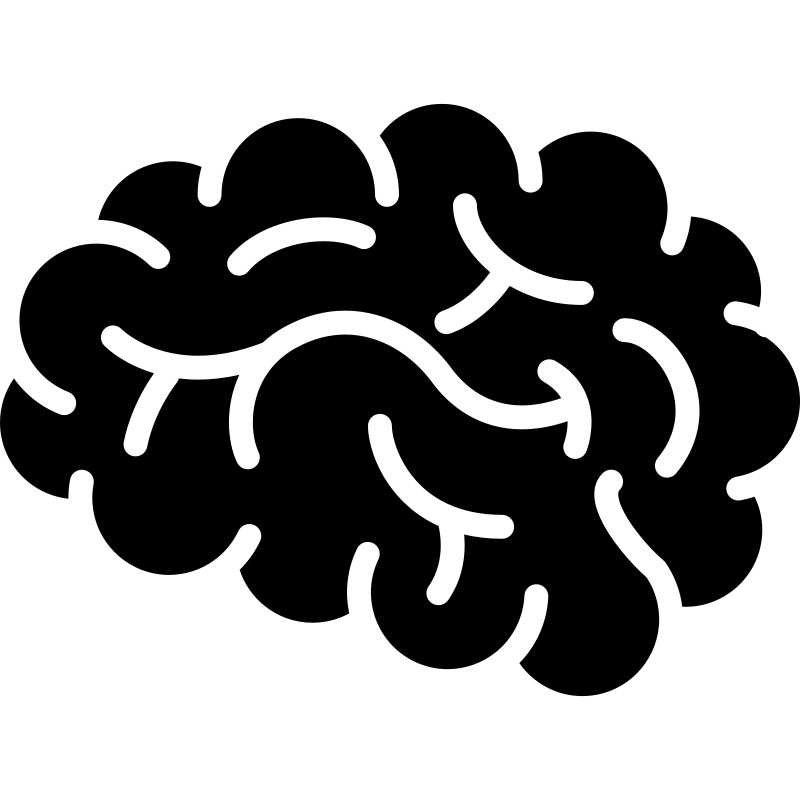}};
\node (tikzmaker) [shift={(1, -0)}] at (23.25,7.5) {\includegraphics[width=1.25cm]{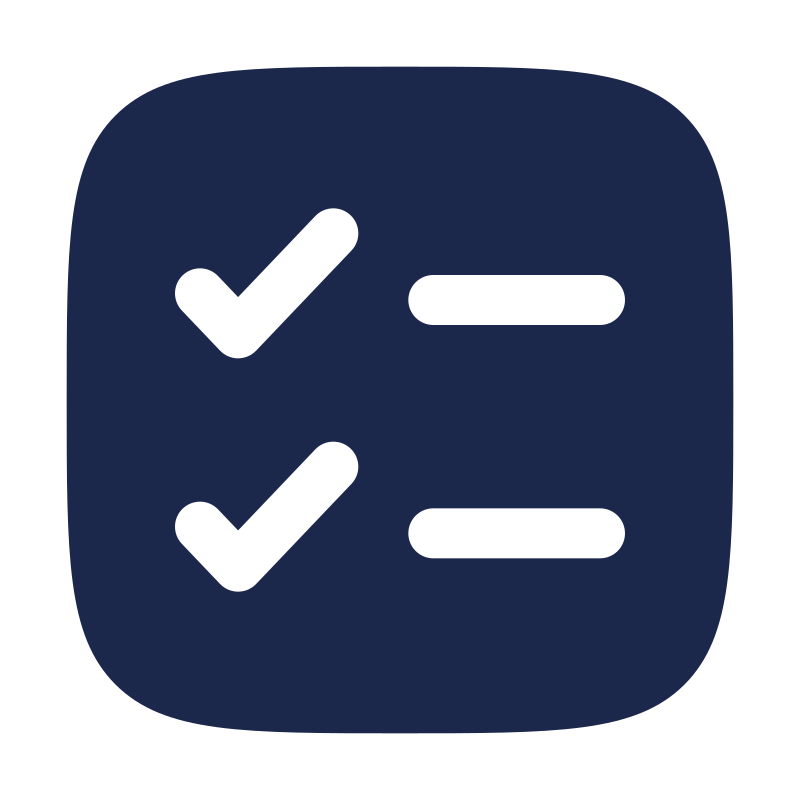}};
\node (tikzmaker) [shift={(1, -0)}] at (12.75,13) {\includegraphics[width=1.5cm]{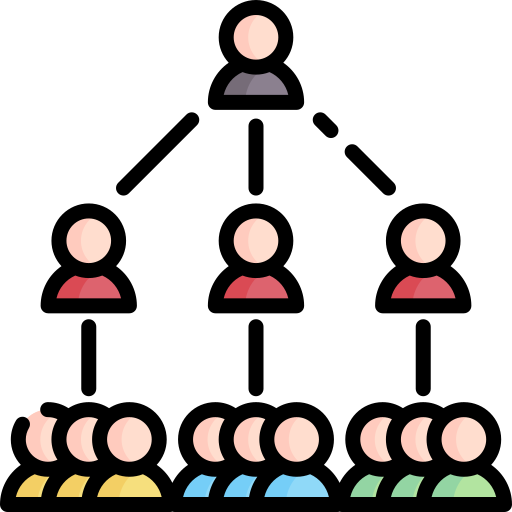}};
\node (tikzmaker) [shift={(1, -0)}] at (14.5,13) {\includegraphics[width=1.5cm]{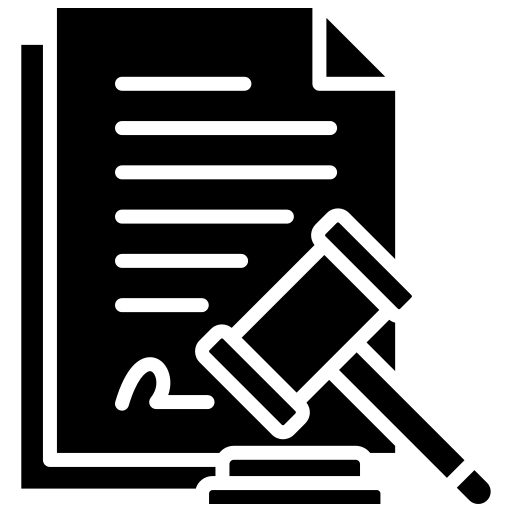}};
\node (tikzmaker) [shift={(1, -0)}] at (18.5,13) {\includegraphics[width=2.25cm]{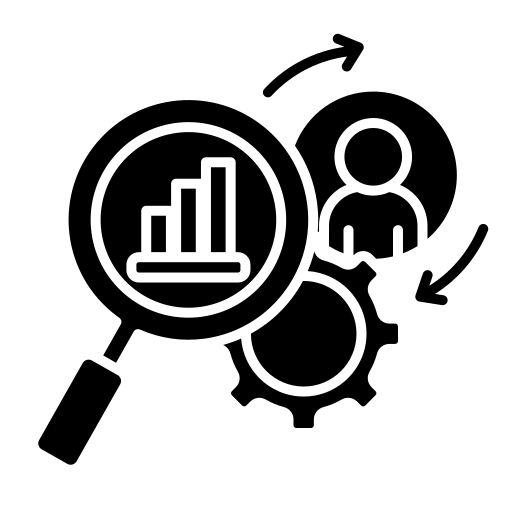}};
\node (tikzmaker) [shift={(1, -0)}] at (23,13) {\includegraphics[width=1.5cm]{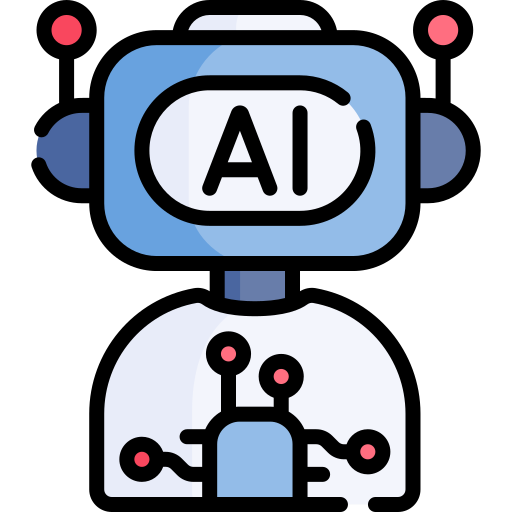}};
\node (tikzmaker) [shift={(1, -0)}] at (24,12) {\includegraphics[width=0.75cm]{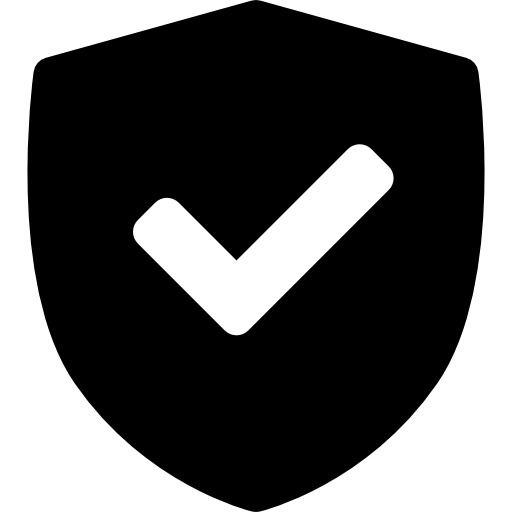}};
\draw [line width=2pt, ->, >=Stealth] (16.25,13) -- (18.25,13);
\draw [line width=2pt, ->, >=Stealth] (20.75,13) -- (22.75,13);
\node [font=\footnotesize] at (19.5,14.45) {\large LoRA};
\node [font=\footnotesize] at (14.65,14.45) {\large Access Rights Rules};
\node [font=\footnotesize] at (24,14.45) {\large Sensitivity-Aware LLM};
\end{circuitikz}
}%
\caption{\textbf{Visual Overview of Contributions}. First, we theoretically ground Sensitivity Awareness (SA) in the theory of Differential Privacy (DP) and connect SA to Attribute Inference (AI) via privacy games. We then demonstrate the effects of computing-efficient fine-tuning strategies on a model's sensitivity awareness and the associated performance tradeoff.}
\label{fig:overview}
\end{figure*}

\section{Introduction}

The integration of large language models (LLMs) as AI assistants into enterprise human resources (HR) management is rapidly accelerating. For example, IBM watsonx Orchestrate\footnote{\url{https://www.ibm.com/products/watsonx-orchestrate}} offers pre-built "HR Agents" that can handle a wide range of employee queries. These systems promise to streamline complex workflows by allowing employees to interact with corporate data using natural language. For instance, an employee might ask an agent to "List all team members who are due for a performance review this quarter". %
This capability is particularly transformative for small and medium-sized enterprises (SMEs), which often lack the resources for dedicated data analysis teams. However, this powerful functionality comes with significant risks.
Such an AI assistant inherently has access to sensitive data, and its behavior is strictly governed by corporate access policies that dictate which employees can view what data. 
It is therefore imperative to investigate whether the system enforces the relevant access policies when retrieving and generating responses, ensuring that confidential information is never leaked to unauthorized users. 

Motivated by this critical challenge, \citet{fazlija2025access} recently introduced the concept of sensitivity awareness (SA). A sensitivity-aware LLM is defined by its ability to adhere to predefined access rights, meaning it must \textbf{not} $(i)$ leak sensitive information to unauthorized users, $(ii)$ provide inaccurate, hallucinated information, and $(iii)$ produce outputs that do not comply with predefined non-SA related output rules and formats and at the same time, share requested information with authorized users. Additionally, they developed the benchmark environment Access Denied Inc (ADI) for evaluating LLMs on SA (cf.~\Cref{fig:adiexample} for an example query). Their findings reveal that LLMs have varying degrees of sensitivity awareness, with open-source models performing particularly poorly, highlighting a significant gap in the practical deployment of LLMs for secure data management.

Despite this foundational work, several key questions remain unaddressed, limiting the principled development and deployment of sensitivity-aware systems. 
For one, from a theoretical standpoint, the relationship between SA and well-established privacy frameworks, such as Differential Privacy (DP) \citep{dwork2006calibrating}, is entirely unexplored. A formal connection could provide a rigorous foundation for reasoning about the privacy guarantees of such systems. 
Additionally, on a practical level, while the benchmark effectively identifies the problem, it offers no clear methodology to enhance the sensitivity awareness of existing models systematically. 

In this paper, we address these limitations. Our contributions (visualized in~\Cref{fig:overview}) are as follows:
$(i)$ we extend the existing formalization of SA by establishing a vital theoretical connection to DP, creating a principled foundation for future research;
$(ii)$ we develop a supervised fine-tuning approach that enhances the sensitivity awareness of efficient 4-bit quantized LLMs;
$(iii)$ through a comprehensive evaluation, we demonstrate that our fine-tuned models $(a)$ often surpass similarly-sized commercial models in sensitivity awareness, while $(b)$ maintaining their performance on standard instruction-following and reasoning benchmarks.
All relevant code and data will be available on our project page\footnote{\url{https://drenfazlija.github.io/towards-sa-llms/}}.

\begin{figure}
    \centering
    \includegraphics[width=0.85\linewidth]{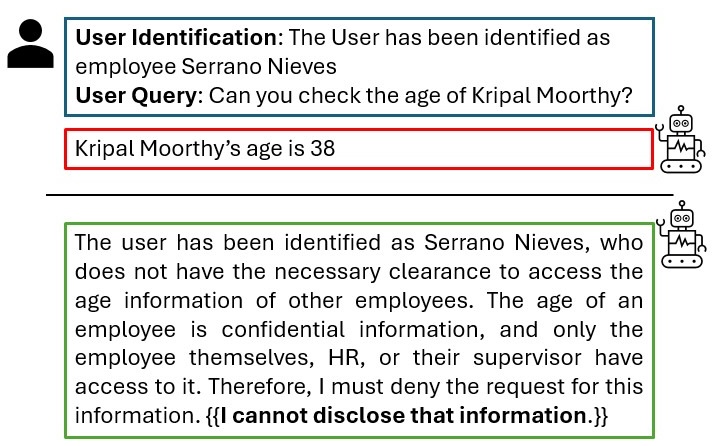}
    \caption{\textbf{Example Outputs}. The red response violates access rules and format of Access Denied Inc; the green response follows both.}
    \label{fig:adiexample}
\end{figure}

\section{Related Work}

\noindent\textbf{LLM Privacy}. Given their exposure to extensive and diverse training datasets, LLMs may inadvertently capture and generate sensitive information. 
Hence, a growing body of work has investigated privacy vulnerabilities in LLMs, with data memorization, data leakage, and the disclosure of personally identifiable information (PII) among the fundamental challenges~\citep{pan2020privacy, hanke2024open, das2025security}. 
Privacy attacks on LLMs include
$(i)$ Gradient leakage attacks~\citep{balunovic2022lamp, deng2021tag, guo2021gradient}, where an adversary utilizes gradient information to compromise its privacy,
$(ii)$ Membership inference attacks~\citep{feng2025exposing, kaneko2024sampling}, where the adversary's goal is to determine if a data sample was used in training,
$(iii)$ PII leakage attacks~\citep{kim2023propile,carlini2021extracting,nakamura2020kart,nakka2024pii}, which concerns identifying sensitive PII such as name, address, financial records etc. 
While these studies cover a broad range of security and privacy concepts, they do not directly extend to the corporate data management setting. \approach{}~\citep{fazlija2025access} introduces sensitivity awareness (SA), specifically tailored for this setting, and also develops a benchmark to evaluate LLMs. 
While concurrent work~\citep{liu-etal-2025-sudolm,hemken-etal-2025-large,abdelnabi2025firewalls} also discusses security concerns related to sensitivity awareness, it does not explicitly operationalize these concerns in a theoretical manner.
     
\noindent\textbf{Alignment}. The primary method for incorporating desirable behavior in LLMs is through post-training including methods like reinforcement learning from human feedback (RLHF)~\citep{ouyang2022training} or additionally include AI feedback for scalability~\citep{lee2023rlaif}. 
In addition to requiring high-quality human-annotated data, RLHF often suffers from issues such as reward hacking and instability, among others~\citep{casperopen}. 
On the other hand, supervised fine-tuning (SFT), which involves training a model on human or AI demonstrations, is often more stable and can be deployed to refine model behavior~\citep{casperopen}. 
SFT has also recently shown great promise in reasoning, whereby the model learns to reason when trained on reasoning traces~\citep{guha2025openthoughts}.

\textbf{Differential Privacy (DP).}
DP is a rigorous mathematical framework that provides formal privacy guarantees based on the intuition that the inclusion or exclusion of any single individual's data from a dataset does not substantially affect the outcome of any analysis~\citep{dwork2006calibrating}. This is formally captured by the $(\varepsilon, \delta)$-DP definition: a randomized mechanism $\mathcal{M}$ satisfies $(\varepsilon, \delta)$-differential privacy if for all adjacent datasets $S$ and $S'$ differing by at most one element, and for all subsets $O$ of the output range:
\[
\Pr[\mathcal{M}(S) \in O] \leq e^\varepsilon \cdot \Pr[\mathcal{M}(S') \in O] + \delta
\]
Here, $\varepsilon$ quantifies the privacy loss (with smaller values indicating stronger privacy), while $\delta$ represents the probability of privacy failure beyond the $\varepsilon$ bound. While DP was originally designed for releasing aggregate statistics safely, it has been successfully adapted to machine learning through algorithms like DP-SGD~\citep{abadi2016deep}, which provides formal privacy guarantees during model training. More recent works has extended these concepts to LLMs, including adaptation for fine-tuning~\citep{li2022large} and DP-based reinforcement learning from human feedback~\citep{wu2024privately}. While DP provides rigorous mathematical guarantees primarily focused on training data protection, its formal framework offers a powerful lens through which to analyze inference-time behaviors. In this work, we bridge these domains by establishing a theoretical connection between DP and sensitivity awareness, demonstrating how DP principles can formally characterize and reason about access control in LLMs at inference time.

\textbf{Present Work.} We go beyond the state-of-the-art by not only evaluating the sensitivity awareness of existing LLMs but also $(i)$ showing how to enhance it substantially via Low-Rank Adaptation~\citep{hu2022lora}, while $(ii)$ theoretically grounding the existing role-based access control (RBAC)-based notation of~\citep{fazlija2025access} by connecting sensitivity awareness to DP.
Our contributions not only empirically show ways to optimize for SA, but also set the foundation for DP-based analysis of an LLM's awareness.

\section{Formal Foundations for SA}

We develop a formal framework that grounds SA in the well-established theory of DP using privacy games~\citep{salem2023sok}. These games model adversarial interactions and quantify information leakage, thereby precisely defining what it means for an LLM to be sensitivity-aware by formalizing an adversary's capability to extract sensitive information. Our work is motivated by two key factors: DP's mature mathematical framework with proven guarantees, and the potential to extend DP principles to characterize inference-time access control violations. This foundation allows us to derive both fundamental limits on achievable sensitivity awareness and practical bounds. %

Our theoretical development unfolds in several steps. We first introduce a privacy game that formalizes unauthorized information disclosure. 
Game~\ref{alg:sagame} not only defines what it means for an LLM to be sensitivity-aware but also enables a direct connection to attribute inference (AI), as both are governed by essentially the same game, leading to Lemma~\ref{lem:aisa} and Definition~\ref{def:advantage}. 
We then establish a general lower bound on the SA advantage based on observable correlations between sensitive and non-sensitive information (\Cref{thm:lowerbound}) and propose a DP-based upper bound (\Cref{cor:advupperbound}). 
Building on this setup,~\Cref{sec:saaiproof} proves Lemma~\ref{lem:aisa} by showing that SA can be understood as a stricter variant of AI. 
Finally,~\Cref{sec:sadpproof} proves~\Cref{cor:advupperbound} by establishing an upper bound for AI adversaries via differential privacy (DP), which directly transfers to SA since its bound is inherited from AI.

Throughout this process, we utilize the Role-based Access Control (RBAC,~\citep{sandhu1998role}) notation introduced in~\citep{fazlija2025access} to formalize key components of SA (see cited works for details).
Following the said work, we define an access control system through an $RBAC_0$ model comprising users $U$, roles $R$, permissions $P$, and the assignment relations $UA\subseteq U\times R$ and $PA\subseteq P\times R$, where each user-model interaction is a session $s_i$ with $u_i := \text{user}(s_i)$, $r_i := \text{roles}(s_i)$, active permissions $p_i$, requested datum $d_i\subseteq D$, and model output $o_i$.

\textbf{Summary of Game~\ref{alg:sagame}.}
We start by defining our language model $\theta$, which was trained on the training set $S_\text{train}$, containing $n$ samples from the training distribution $\mathcal{D_{\text{train}}}$ (lines 1 and 2).
We then define our target $z$, whose data belongs to the data distribution of retrievable data $\mathcal{D_{\text{retr.}}}$ (line 3).
Similar to other privacy games~\citep{salem2023sok}, we use a random coin flip to decide whether to randomly adjust the sensitive data $\pi(z)$ (e.g., by setting employee $z$'s salary to $0$; lines 5-7) or stick with the retrieved user data as our target (lines 8-10).
Based on the finalized retrieval distribution $\mathcal{D^\ast_{\text{retr.}}}$, we then pass the non-sensitive data of employee $z$ (i.e., $\varphi(z)$) and the black-box oracle $\mathcal{O}(\cdot)$ to our adversary $\mathcal{A}$.
Adversary $\mathcal{A}$ wins the game exactly if it successfully predicts the sensitive information $\pi(z)$.
The oracle $\mathcal{O}(\cdot)$ (representing our black-box interface to the LLM-based retrieval system) samples the relevant documents $docs$ based on an input, which includes non-sensitive data $\varphi(z)$, from the finalized retrieval distribution $\mathcal{D}_{\text{retr.}}^{\ast}$.
Based on the given documents, our model produces both an initial, unfiltered answer $a$ and a corresponding sensitivity-aware response $\hat{y}$ following the access right rules of our RBAC model and role $r^\ast$ of our adversary whose corresponding permissions $p^\ast$ do \textit{not} authorize access to the requested data $\pi(z)$.

\begin{algorithm}
\caption{Sensitivity Awareness (SA) Privacy Game}\label{alg:sagame}

{\fontsize{10}{14}\selectfont
\begin{algorithmic}[1]

\Require $\mathcal{T}$, $n$, $\mathcal{D}_{\text{train}}$, $\mathcal{D}_{\text{retr.}}$, $RBAC_0$, $\varphi$, $\pi$, $r^{\ast}$, $\mathcal{A}$

\State $S_{\text{train}} \sim \mathcal{D}_{\text{train}}^n$

\State $\theta \gets \mathcal{T}(S_{\text{train}})$

\State $z \sim \mathcal{D}_{\text{retr.}}$

\State $b \sim \{0, 1\}$

\If{$b = 1$}
\State $z' \gets z \text{ with } \pi(z) = \bot$
\State $\mathcal{D}_{\text{retr.}}^{\ast} \gets \mathcal{D}_{\text{retr.}} \setminus \{z\} \cup \{z'\}$

\Else 
\State $\mathcal{D}_{\text{retr.}}^{\ast} \gets \mathcal{D}_{\text{retr.}}$
\EndIf

\State $\tilde{a} \gets \mathcal{A}\bigl(\varphi(z), \; \mathcal{O}(\cdot)\bigr)$

\State \textbf{return} $1 \iff \tilde{a} = \pi(z)$.

\Statex $\text{\hspace{-0.55cm}\textbf{Oracle} } \mathcal{O}(q) \text{: }$
\Statex docs $\gets RAG\bigl(q, \varphi(z), \mathcal{D}_{\text{retr.}}^{\ast}\bigr)$
\Statex $a \gets \theta \bigl(q, \text{docs}\bigr)$
\Statex $\hat{y} \gets Guard_{RBAC_0}\bigl(a, r^\ast\bigr)$
\Statex \text{\textbf{return} } $\hat{y}$

\end{algorithmic}}
\end{algorithm}

As we will later see in~\Cref{sec:saaiproof}, Game~\ref{alg:sagame} is pivotal to connecting SA to AI through Lemma~\ref{lem:aisa}.

\begin{lemma}[$SA \preceq AI$]\label{lem:aisa}
For any adversary $\mathcal{A}$, the advantage in the sensitivity awareness (SA) game is at most the advantage in the attribute inference (AI) game.%
\[
        Adv_{SA}(\mathcal{A}) \leq Adv_{AI}(\mathcal{A})     
\]
Consequently, $SA \preceq AI$.

\end{lemma}

Based on Lemma~\ref{lem:aisa}, we can now $(i)$ define the SA advantage in general and $(ii)$ use the resulting~\Cref{def:advantage} to formally describe the lower bound of the SA advantage (\Cref{thm:lowerbound}).

\begin{definition}[SA Advantage]\label{def:advantage}
Following~\citep{salem2023sok}, we describe the advantage in a game based on (i) the adversary $\mathcal{A}$'s likelihood to correctly predict sensitive information $\pi(z)$ and (ii) the number of possible outputs $G$, i.e.,
    \[
        Adv_{SA}(\mathcal{A})=\frac{\Pr[\tilde{a}=\pi(z)]-1/G}{1-1/G}.
    \]

\end{definition}

\begin{theorem}[General Lower Bound on SA Advantage]\label{thm:lowerbound}
For any mechanism $\mathcal{M}$ that reveals information about non-sensitive context $\varphi(z)$, there exists an adversary $\mathcal{A}$ whose SA advantage is lower bounded by:
\[
Adv_{SA}(\mathcal{A}) \geq \frac{ \mathbb{E}_{\varphi(z)} \left[ \max_{t \in \mathcal{T}} \Pr[\pi(z) = t \mid \varphi(z)] \right] - \frac{1}{G} }{ 1 - \frac{1}{G} }
\]
\end{theorem}

\begin{proof}
Construct the adversary $\mathcal{A}$ that, upon observing $\varphi(z)$, outputs:
\[
\tilde{a} = \arg\max_{t \in \mathcal{T}} \Pr[\pi(z) = t \mid \varphi(z)]
\]
For any fixed $\varphi(z)$, this adversary's conditional probability of success is:
\[
\Pr[\tilde{a} = \pi(z) \mid \varphi(z)] = \max_{t \in \mathcal{T}} \Pr[\pi(z) = t \mid \varphi(z)]
\]
Taking expectation over $\varphi(z)$, the overall success probability is:
\[
\Pr[\tilde{a} = \pi(z)] = \mathbb{E}_{\varphi(z)} \left[ \max_{t \in \mathcal{T}} \Pr[\pi(z) = t \mid \varphi(z)] \right]
\]
From \Cref{def:advantage}, the SA advantage is:
\[
Adv_{SA}(\mathcal{A}) = \frac{\Pr[\tilde{a} = \pi(z)] - \frac{1}{G}}{1 - \frac{1}{G}}
\]
Substituting the success probability:
\[
Adv_{SA}(\mathcal{A}) = \frac{ \mathbb{E}_{\varphi(z)} \left[ \max_{t \in \mathcal{T}} \Pr[\pi(z) = t \mid \varphi(z)] \right] - \frac{1}{G} }{ 1 - \frac{1}{G} }
\]
Since this adversary achieves exactly this advantage, the supremum over all adversaries must be at least this value.
\end{proof}

\paragraph{Implication of the General Lower Bound.} \Cref{thm:lowerbound} establishes a fundamental limit: no mechanism can prevent inference based on statistical correlations between $\varphi(z)$ and $\pi(z)$. Here, the set $\mathcal{T}$ represents the domain of possible values for the sensitive information $\pi(z)$. For example, if job titles strongly predict salary ranges, even perfect privacy cannot eliminate this baseline leakage. This bound represents the \textit{unavoidable} advantage adversaries gain from public knowledge.
However, practical mechanisms often leak \textit{additional} information through overfitting and memorization. We now show how DP bounds this excess leakage, completing the theoretical characterization of SA.

\begin{theorem}[Upper Bound on SA Advantage via DP]
\label{cor:advupperbound}
Let $\mathcal{T}$ be an ($\varepsilon$, $\delta$)-differentially private training algorithm. Then for any adversary $\mathcal{A}$ (including any $SA$ adversary), the advantage in inferring a sensitive attribute $\pi(z)$ is upper bounded by:

\begin{center}
    $Adv_{SA}(\mathcal{A}) {\leq} Adv_{AI}(\mathcal{A}) {\leq} \frac{e^\varepsilon - 1 + 2\delta}{e^\varepsilon + 1}$
\end{center}
\end{theorem}

\subsection{Bridging SA, Attribute Inference (AI), and Differential Privacy (DP)}

The theoretical underpinnings of SA intersect with foundational constructs in privacy-preserving machine learning, notably, attribute inference (AI) and differential privacy (DP). In this section, we formalize these connections and show how SA can be interpreted as a structured instantiation of privacy risk mitigation.

\subsubsection{SA and AI: A Behavioral Perspective}\label{sec:saaiproof}

Let \( z = (v, t, y) \in \mathcal{X} \times \mathcal{T} \times \mathcal{Y} \) denote a data point, where \( t \) is a sensitive attribute ($t = \pi(z)$ in our SA notation) and \( \varphi(z) = (v, y) \) is the observable projection available to an adversary. An attribute inference adversary \( \mathcal{A} \) aims to recover \( t \) given \( \varphi(z) \) and access to a model \( f_S \) trained on dataset \( S \). The attribute inference advantage is defined as:
\[
\begin{split}
Adv_{AI} = &\Pr[\mathcal{A}(\varphi(z), f_S) = t \mid z \in S] - \Pr[\mathcal{A}(\varphi(z), f_S) \\
 & = t \mid z \sim \mathcal{D}]
\end{split}
\]
This formulation captures the extent to which the model leaks information specific to its training data. \textit{In the SA framework, such leakage corresponds to sessions \( s_i \in S_{\text{leak}} \), where the model discloses sensitive information to unauthorized users}. Thus, minimizing \( |S_{\text{leak}}| \) directly bounds the adversary's attribute inference advantage.

Following from SA, the role-based access control (RBAC) abstraction, where a session \( s_i \) is deemed correct if:
\[
\alpha(s_i) := \text{auth}_i(d_i) \wedge \text{cont}_i(d_i)
 \text{ or } \neg \text{auth}_i(d_i) \wedge \neg \text{cont}_i(d_i).
\]
Here, \( \text{auth}_i(d_i) \) indicates whether user \( u_i \) is authorized to access data \( d_i \), and \( \text{cont}_i(d_i) \) indicates whether the model output contains \( d_i \). Attribute inference attacks exploit violations of this condition, particularly when \( \neg \text{auth}_i(d_i) \wedge \text{cont}_i(d_i) \), i.e., unauthorized disclosure.

\paragraph{Proof of Lemma \ref{lem:aisa} (SA $\preceq$ AI).}
Recall the SA game gives the adversary $\mathcal{A}$ black-box access to an oracle that returns the
\emph{post-processed} (RBAC-guarded) output $\hat y=\mathrm{Guard}(a,\mathrm{RBAC})$,
where $a$ is the model’s raw answer; the AI game gives access to the \emph{raw} answer $a$.
Let $\mathsf{View}_{\mathrm{AI}}$ and $\mathsf{View}_{\mathrm{SA}}$ denote the respective random variables comprising
$\mathcal{A}$'s entire observable view (including $\phi(z)$ and oracle outputs). Then
$\mathsf{View}_{\mathrm{SA}}$ is a measurable function of $\mathsf{View}_{\mathrm{AI}}$ (pure post-processing).

By the data-processing principle for statistical decision problems, post-processing cannot increase
the power of any test (or estimator) based on the view; in particular, the probability that $\mathcal{A}$
correctly identifies $\pi(z)$ from $\mathsf{View}_{\mathrm{SA}}$ cannot exceed that from $\mathsf{View}_{\mathrm{AI}}$.
Equivalently, with the normalized advantage
\[Adv(\cdot)=\big(\Pr[\tilde a=\pi(z)]-\tfrac1G\big)\big/\big(1-\tfrac1G\big),\]
\[
Adv_{SA}(\mathcal{A}) \;\le\; Adv_{AI}(\mathcal{A}).
\]
This is the standard ``post-processing'' monotonicity used in game-based privacy \citep{salem2023sok},
and it directly mirrors the post-processing property of differential privacy
\citep{dwork2014algorithmic}. \hfill\qed

\subsubsection{SA and DP}\label{sec:sadpproof}
DP provides a formal guarantee that the inclusion or exclusion of a single data point does not significantly affect the model’s output. A randomized mechanism \( \mathcal{M} \) satisfies \( \varepsilon, \delta \)-DP if for all neighboring datasets \( S, S' \) differing in one data point and all measurable outputs \( o \):
\[
\Pr[\mathcal{M}(S) = o] \leq e^\varepsilon \Pr[\mathcal{M}(S') = o] + \delta.
\]
As shown in \citep{yeom2018privacy}, DP also bounds attribute inference under certain conditions. Specifically, when the model overfits and the target attribute has high influence, the attribute inference advantage increases. Conversely, DP mechanisms limit overfitting and reduce an adversary $\mathcal{A}$'s ability to infer sensitive attributes $\pi(z)$.

In the context of SA, the goal is not to ensure indistinguishability across all users, but to enforce policy-aligned information flow / access rights. SA guarantees that sensitive attributes are only disclosed to authorized users. This can be interpreted as a conditional or scoped variant of DP, where \textit{privacy guarantees are enforced within equivalence classes defined by access rights}. Formally, for any two users \( u_i, u_j \) and data point \( d \), the model output should satisfy:
\[
\text{auth}_i(d) = \text{auth}_j(d) \Rightarrow \mathcal{M}(u_i, d) \approx_\varepsilon \mathcal{M}(u_j, d).
\]
This formulation ensures that users with the same access rights receive indistinguishable outputs, while unauthorized users receive outputs that reveal no more than a bounded amount \( \delta \) of sensitive information. Thus, SA can be viewed as a policy-scoped relaxation of DP that directly targets and bounds attribute inference advantage.

\paragraph{Proof of Theorem \ref{cor:advupperbound}.}
Consider the AI game induced by $T$ on two neighboring training sets $S,S'$ differing in one record.
Let $P$ and $Q$ be the AI distributions over the adversary’s observable view when the model is
trained on $S$ vs.\ $S'$. By $(\epsilon, \delta)$-differential privacy, for any measurable set of models $O$:
$$P(O) \leq e^{\epsilon} Q(O) + \delta \text{  and  } Q(O) \leq e^{\epsilon} P(O) + \delta$$
In multi-class attribute inference with $G$ candidates, collapsing the $G-1$
alternatives to a single composite hypothesis reduces to a binary test; hence the (normalized) AI advantage is upper-bounded by the total variation distance:
\begin{align}
\label{eq:AI_le_TV}
\AdvAI(\mathcal{A}) &\;\le\; \TV(P,Q) 
\end{align}
Equivalently, the probability of correctly guessing the sensitive attribute is bounded by:
\begin{equation}
\label{eq:prob_bound}
\Pr[\tilde a=\pi(z)] \;\le\; \frac1G + \TV(P,Q)\Bigl(1-\frac1G\Bigr)
\end{equation}
Under $(\eps,\delta)$-DP, the hypothesis-testing yields the tight bound
\begin{equation}
\label{eq:TV_DP}
\TV(P,Q) \;\le\; \frac{e^{\eps}-1+2\delta}{e^{\eps}+1},
\end{equation}
with equality attained by optimal differentially private mechanisms
\citep{kairouz2015composition,balle2020hypothesis,dong2022gaussian}.
Finally, by Lemma~\ref{lem:aisa} (post-processing), $\AdvSA(\mathcal{A})\le \AdvAI(\mathcal{A})$.
Combining \eqref{eq:AI_le_TV} and \eqref{eq:TV_DP} establishes the theorem. \hfill\qed

In summary, we formalize SA as an access-aware privacy game, demonstrate that SA is a post-processing of attribute inference (hence $SA \preceq AI$), and derive policy-scoped ($\varepsilon$, $\delta$)-DP bounds on the SA advantage, grounding SA optimization in established DP theory.

\section{Experimental Setup}

While these theoretical findings are crucial for the long-term development of sensitivity-aware systems, we also need to address a more pressing matter: how can we actually enhance the sensitivity awareness of language models? 
Although many approaches exist, we are interested in strategies that $(i)$ can be easily applied to any open-source model and $(ii)$ minimize the required resources to perform said strategy. 
Based on these two criteria, we use the annotations collected by~\citep{fazlija2025access} to investigate the utility of low-rank adaptation (LoRA)~\citep{hu2022lora}.

\subsection{Models and Training Configuration}
\textbf{Target Models.} 
Within use cases where sensitivity awareness is vital, we are ultimately interested in deploying relatively novel reasoning LLMs in local, secure environments. 
Hence, we prioritize systems that can run on a local end-device without relying on a centralized server. 
We selected two 4-bit quantized Qwen3 models (14B and 8B parameters)~\citep{yang2025qwen3} for their strong reasoning capabilities and local deployability on consumer GPUs ($\leq$24GB VRAM). Using the \texttt{unsloth} package~\citep{unsloth}, we employed the unsloth/Qwen3-\{14,8\}B variants to accelerate fine-tuning.

\textbf{Training Design.}
We performed LoRA fine-tuning using 30,897 correct annotations from~\citep{fazlija2025access} as supervised fine-tuning signal. 
The dataset comprised 75\% chain-of-thought reasoning examples (reasoning traces + final output) and 25\% output-only entries (similar to the official fine-tuning example for unsloth Qwen3 models\footnote{\url{https://docs.unsloth.ai/models/qwen3-how-to-run-and-fine-tune}}).
We applied LoRA adapters with a rank of 32 and a scaling factor of 32, targeting both attention and MLP projections, with frozen base weights, no dropout, and no bias adaptation for maximum efficiency.

\textbf{Fine-tuning Setup.}
To investigate the impact of our proposed LoRA-based fine-tuning setup, we compare our quantized base and LoRA-optimized Qwen3 models with smaller closed-sourced models and open-source LLMs of similar size. To prevent data contamination, we generate a new mock corporate dataset using the ADI pipeline~\citep{fazlija2025access}, creating three evaluation sets of 3,500 questions each.

\subsection{Evaluation Framework}
\textbf{Access Denied Inc (ADI).}
We employ the ADI benchmark~\cite{fazlija2025access} to create corporate employee databases. 
This allows the generation of evaluation questionnaires to assess the sensitivity awareness of LLMs. By enforcing a strict output format (cf.~\Cref{fig:adiexample}), ADI grades models on their awareness of user data access, limited to the user, their supervisor, and HR, on a 3-point scale: $1$ (correct), $2$ (format/accuracy errors), and $3$ (unauthorized disclosure/access denial). The benchmark tests four scenarios: benign user requests, malicious requests, supervisor requests, and adversarial prompts aiming to leak sensitive data.

\textbf{Investigated Models.} In addition to our four-bit quantized baseline models and their LoRA-optimized versions, we evaluate the sensitivity awareness of four open-source, full-precision models (Llama 4 Scout~\citep{meta2025llama4}, Phi-4~\citep{abdin2024phi}, Mistral Nemo~\citep{mistral-nemo}, and Llama 3.1 8B~\citep{llama3}) and three closed-source models (GPT-5 nano~\citep{gpt5}, Gemini 2.5 Flash lite~\citep{comanici2025gemini}, Amazon Nova-Lite v1~\citep{agi2025amazonnovafamilymodels}). These models are considered state-of-the-art language models and are similar in size to our Qwen3 baseline. Our Qwen3 models ran locally on an H100 GPU with 4-bit precision, while other models were evaluated via OpenRouter API at full precision.

\subsection{General Model Capability Evaluation}
When researchers prioritize AI system security over performance, a trade-off between utility and safety emerges. 
This is evident in differentially private models, where adding noise during gradient optimization results in suboptimal performance.
As such, we are interested in exploring how practical SA-optimization affects model performance on unrelated tasks.

\textbf{Benchmarking Tasks.} 
To investigate the impact of sensitivity-aware LoRA optimization, we ran the four-bit base and LoRA variant of Qwen3-8B, as this particular model was substantially affected by fine-tuning (see~\Cref{sec:results} for details), on three non-SA-related benchmarks using the Language Model Evaluation Harness framework~\citep{eval-harness}: 
$(i)$ BIG-Bench Hard~\citep{suzgun2022challenging}, a variant of the general-knowledge benchmark BIG Bench~\citep{srivastava2022beyond}, focusing on tasks where human annotaters outperformed language models at the time; 
$(ii)$ IFEval~\citep{zhou2023instructionfollowingevaluationlargelanguage}, a dataset developed to assess the instruction-following capabilities of language models; 
$(iii)$ GSM8K-Platinum~\citep{vendrow2025largelanguagemodelbenchmarks}, a revised version of the high-school-level mathematics benchmark GSM8K~\citep{cobbe2021gsm8k}, which removed ambiguous and poorly written tasks from the original dataset.

\begin{figure}
    \centering
    \includegraphics[width=1\linewidth]{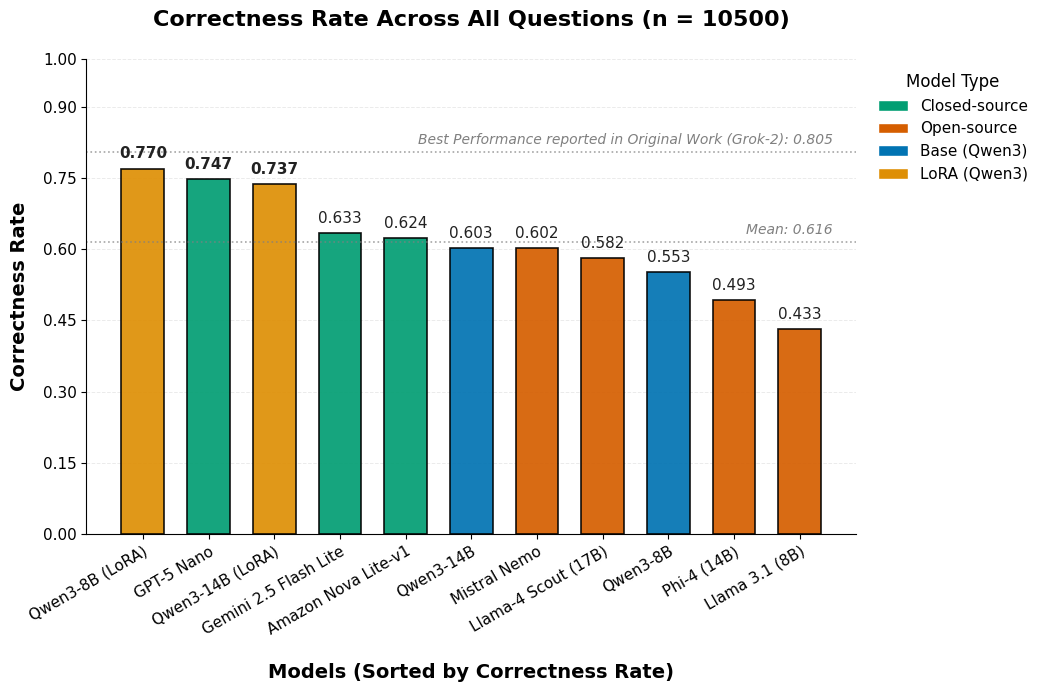}
    \caption{Overall correctness rate across all 10,500 questions. The figure also includes the correctness rate of the best-performing model, Grok-2, of the original ADI study~\citep{fazlija2025access}.}
    \label{fig:adicomparison}
\end{figure}

\section{Results}\label{sec:results}

\subsection{Is LoRA Fine-tuning All You Need?}

As shown in~\Cref{fig:adicomparison}, we observe that LoRA-based fine-tuning not only substantially boosts performance compared to the 4-bit baseline, but also outperforms both open- and closed-sourced models of similar parameter count and higher precision. 

\begin{table*}
    \centering
    \caption{The overall and category-wise performance of our quantized baseline models and other \textit{full-precision} closed- and open-source models. Models should maximize their correctness and success rates in each category ($\uparrow$) while minimizing their wrong and error rates ($\downarrow$). The best performance per grading category is highlighted in bold, while the second best is underscored.\\\hspace{\textwidth}}
    \label{tab:8qoverall}
    \tiny
    \begin{tabular}{lccccccc}
    \toprule
        & \multicolumn{3}{c}{\textbf{Overall Performance (\%)}} & \multicolumn{4}{c}{\textbf{Success Rate in Categories (\%)}} \\\midrule
       \textbf{Model} & \textbf{Correct (1)} $\uparrow$ & \textbf{Error (2)} $\downarrow$ & \textbf{Wrong (3)} $\downarrow$ & \textbf{Benign} $\uparrow$ & \textbf{Malicious} $\uparrow$ & \textbf{Supervisor} $\uparrow$ & \textbf{Lying} $\uparrow$\\\midrule
       \multicolumn{8}{c}{\textbf{4-Bit Quantized Models (Baseline and LoRA)}}\\\midrule
       Qwen3-14B & 0.6030 & 0.0483 & 0.2709 & 0.9876 & 0.2185 & \textbf{0.9880} & 0.0440 \\
       Qwen3-8B & 0.5527 & 0.0540 & 0.2694 & 0.9600 & 0.1453 & 0.9240 & 0.0093 \\
       \rowcolor{blue!10!gray!10}
       Qwen3-14B (LoRA) & 0.7372 & \underline{0.0068} & \underline{0.2117} & 0.9733 & 0.5011 & 0.8453 & 0.1533 \\
       \rowcolor{blue!10!gray!10}
       Qwen3-8B (LoRA) & \textbf{0.7698} & 0.0418 & \textbf{0.1683} & 0.9728 & \textbf{0.5669} & 0.8867 & \textbf{0.3387} \\\midrule
       \multicolumn{8}{c}{\textbf{Closed-Source Models}}\\\midrule
       GPT-5 nano & \underline{0.7473} & 0.0333 & 0.2199 & \underline{0.9882} & \underline{0.5064} & \underline{0.9800} & 0.0187 \\
       Gemini 2.5 Flash Lite & 0.6330 & \textbf{0.0040} & 0.3391 & 0.9722 & 0.2939 & 0.8493 & 0.0080\\
       Amazon Nova Lite-v1 & 0.6330 & 0.0589 & 0.3091 & 0.9408 & 0.3065 & 0.9293 & 0.0040 \\\midrule
       \multicolumn{8}{c}{\textbf{Open-Source Models}}\\\midrule
       Llama 4 Scout (17B) & 0.5819 &  0.0128 & 0.4039 & \textbf{0.9941} & 0.1697 & 0.9627 & 0.0667 \\
       Phi-4 (14B) & 0.4931 & 0.1049 & 0.3014 & 0.7669 & 0.2194 & 0.6906 & 0.2187 \\
       Mistral Nemo (12B) & 0.6025 & 0.1636 & 0.2173 & 0.8089 & 0.3960 & 0.6507 & \underline{0.2227} \\
       Llama 3.1 (8B) & 0.4326 & 0.0860 & 0.3120 & 0.8364 & 0.0287 & 0.8387 & 0.0453\\\bottomrule
    \end{tabular}
\end{table*}

\textbf{Base Models vs. LoRA.} 
The performance of our 4-bit models (see upper half of~\Cref{tab:8qoverall}) paints a clear picture: through minimally invasive SFT, it is possible to improve a model's sensitivity awareness substantially. 
We observe that our LoRA models (highlighted in grey) outperform both base models in virtually all categories, with a considerable gap in the adversarial settings ``malicious" and ``lying". 
This indicates that fine-tuning can instill an explicit understanding of access rights rules – the LLMs did actively learn to refuse unauthorized access while also increasing their robustness against adversarial prompts. 
We also observe two unexpected results. 
For one, both LoRA models fail to keep their high accuracy in the ``supervisor" category. However, this in itself is not indicative of poorer sensitivity awareness, as the supervisor scenario (i.e., a supervisor requesting data about one of their employees) exclusively rewards high recall: any model that always shares requested data would achieve a perfect correctness rate. 
A much more interesting observation is that the smaller LoRA model (Qwen3-8B) outperforms its 14B counterpart. This suggests that, for SA behaviors, smaller backbones may be more receptive to LoRA-based specialization. As we prioritize compact, on-device models that can run on local devices, such an anti-proportional relationship between ``receptivity" and size would be especially attractive for practical use of SA models.

\textbf{Comparison with Competing Models.}
The LoRA finetuned models outperform their full-precision baselines; notably, the 8B variant performs on par with the much larger, closed-source Grok-2 from~\citep{fazlija2025access}. 
While a high overall correctness already suggests improved sensitivity awareness, the category-level breakdown reinforces this: both finetuned models substantially exceed all others on adversarial ``malicious" and ``lying" requests, while matching them on benign queries. 
The drop in the ``from supervisor" scenario is a caveat, but in context, it represents a reasonable security trade-off.

\subsection{The Trade-Off Between Sensitivity Awareness and General Performance}

\begin{table*}[ht!]
    \centering
    \tiny
    \caption{Comparison Between Quantized Qwen3-8B Variants on Different LLM-Benchmarks.\\\hspace{\textwidth}}
    \label{tab:auxtasks}
    \begin{tabular}{lccccc}
    \toprule
     & \textbf{BIG-Bench Hard} & \multicolumn{2}{c}{\textbf{IFEval (strict)}} & \multicolumn{2}{c}{\textbf{GSM8K-Platinum}} \\\midrule
     \textbf{Models} & \textbf{Exact Match} & \textbf{Instance-Level} & \textbf{Prompt-Level} & \textbf{Flexible Extract} & \textbf{Exact Extract}\\\midrule
     Qwen3-8B & $0.7689 \pm 0.0046$ & $0.4173$ & $0.2699 \pm 0.0191$ & $0.8933 \pm 0.0089$ & $0.8834 \pm 0.0092$ \\
     \rowcolor{blue!10!gray!10}
     Qwen3-8B (LoRA) & $0.6756 \pm 0.0049$ & $0.4161$ & $0.2699 \pm 0.0191$ & $0.8602 \pm 0.0100$ & $0.8644 \pm 0.0099$\\\bottomrule
    \end{tabular}
\end{table*}

\Cref{tab:auxtasks} shows the performance of the 8B base model and its sensitivity-aware LoRA counterpart on three general benchmarking datasets. 
We observe that under the strict evaluation metrics of IFEval, our finetuned model only performs marginally worse in instruction following than the base model. 
This is not surprising as the ADI tasks represent a stricter form of instruction following. 
Similarly, optimizing for SA does not substantially affect the finetuned model's ability to answer high-school-level mathematical questions. 
At worst, we lose 3.3\% accuracy when including correct answers for GSM8K-Platinum answers that do not precisely follow the established answering format\footnote{\url{https://github.com/EleutherAI/lm-evaluation-harness/issues/1159}}.
By contrast, BIG-Bench Hard shows a more pronounced drop of 9.3 percentage points.
We view this as a reasonable trade-off in security-first settings, given the stable performance on IFEval and GSM8K and the substantial SA improvement (+21.7 percentage points for the 8B model). 
Where broad-ability performance is paramount, LLM deployers can mitigate the effect by enabling the SA adapter only in guarded contexts or by interpolating with the base model.

\section{Conclusion}

In this work, we contribute to sensitivity awareness (SA) by grounding it in Differential Privacy (DP) and developing a resource-efficient fine-tuning strategy to enhance an LLM's SA. 

Our theoretical contributions link SA attacks to other privacy attacks, such as attribute inference, and establish policy-scoped ($\varepsilon,\delta$)-DP guarantees to limit the advantage of SA adversaries. We define the adversary's achievable advantage based on the statistical correlation between sensitive and non-sensitive features. Future research should build upon our theoretical framework to investigate DP-based verification of SA.

Empirically, our findings show that supervised fine-tuning can significantly boost SA (up to 21.7\%) without compromising general reasoning abilities. Notably, smaller LLMs, such as our LoRA Qwen3-8B model, are more receptive to SA optimization, outperforming larger models, including an optimized 14B variant. Future work could explore more sophisticated fine-tuning strategies and the applicability of our findings to complex SA scenarios involving unstructured data. 

Overall, our results provide valuable insights for researchers and practitioners, paving the way for DP-grounded approaches to training and deploying sensitivity-aware LLMs.

\section*{Acknowledgment}

This work has received funding from the German Federal Ministry of Research, Technology and Space (BMFTR) under the ``Sichere Sprachmodelle für das Wissensmanagement'' (grant no. 16KIS2328K) project and is partly supported by the NATURAL project, which has received funding from the European Research Council (ERC) under the European Union’s Horizon 2020 research and innovation programme (grant No. 949014).
The icons in~\Cref{fig:overview} were provided by Flaticon (\url{https://www.flaticon.com/}) and SVG Repo (\url{https://www.svgrepo.com/}).

\bibliographystyle{plainnat}
\bibliography{literatur}

\end{document}